\documentclass[a4paper,12pt]{article}

\usepackage{amsmath,amssymb,amsthm}
\usepackage[margin=24mm]{geometry}
\usepackage[doublespacing]{setspace}
\usepackage{enumerate}
\usepackage[round]{natbib}

\newcommand{\keywords}[1]{\par\noindent\textbf{Keywords:}%
 #1\par}
\newcommand{\jel}[1]{\par\noindent\textbf{JEL classification codes:}%
 #1\par}
\newcommand{\rnd}[1]{\left( #1 \right)}
\newcommand{\crly}[1]{\left\{ #1 \right\}}
\newcommand{\sqr}[1]{\left[ #1 \right]}

\newcommand{\rndsqr}[1]{\left( #1 \right]}\newcommand{\sqrrnd}[1]{\left[ #1 \right)}
\newcommand{\set}[2]{\left\{ #1 : #2 \right\}}
\newcommand{\map}[3]{#1 : #2 \to #3}
\newcommand{\card}[1]{\left| #1 \right|}
\newcommand{\ch}{\mathrm{CH}}
\newcommand{\sh}{\mathrm{SH}}
\newtheorem{proposition}{Proposition}
\newtheorem{theorem}{Theorem}
\newtheorem{lemma}{Lemma}

\newtheorem{axiom}{Axiom}

\theoremstyle{definition}
\newtheorem{definition}{Definition}

\theoremstyle{remark}

\title
{
 Strategies in Deterministic Totally-Ordered-Time Games\thanks
 {
  This version is extreee$\cdots$eeemely preliminary.
 }
}
\author
{
 Tomohiko Kawamori\thanks
 {
  Faculty of Economics, Meijo University,
  1-501 Shiogamaguchi, Tempaku-ku, Nagoya 468-8502, Japan.
  {\tt kawamori@meijo-u.ac.jp}
 }
}
\date
{
}

\begin{document}

\maketitle

\abstract
{
 We consider deterministic totally-ordered-time games.
 We present three axioms for strategies.
 We show that
 for any tuple of strategies that satisfy the axioms,
 there exists a unique complete history that is consistent with the strategy tuple.
}

\keywords
{
 deterministic totally-ordered-time game;
 strategy;
 unique existence of consistent complete history;
 well-ordered set
}

\jel
{
 C72;
 C73
}

\newpage
\section{Introduction}\label{sec:introduction}

In continuous-time games,
some strategy tuples may fail to induce a unique complete history (path of play).
Some strategy tuples may induce no complete history
(see Example 1 in \cite{Kamada2021}).
Some strategy tuples may induce multiple complete histories
(see Example 2 in \cite{Kamada2021}).
Thus,
each strategy tuple does not necessarily specify a unique payoff tuple.

Several papers proposed restrictions imposed on strategies in order that
each strategy tuple induces a unique complete history.
Any strategy $\sigma_i$ of any player $i$ is restricted in the literature as follows.
In \cite{Bergin1993},
in any subgame,
there exists a small initial interval during which player $i$ does not change his/her action in any complete history consistent with strategy $\sigma_i$
(\emph{inertiality}).
In \cite{Kamada2021},
in any subgame,
given any path of the other players' action tuples,
there exists a complete history consistent with $\sigma_i$
(\emph{traceability}),
and
in any subgame,
player $i$ moves only finite times during any finite-length interval in any complete history consistent with $\sigma_i$
(\emph{frictionality}).
Inertiality does not necessarily imply traceability and frictionality,
and
vice versa.

This paper proposes a restriction on strategies that
makes each strategy tuple to induce a unique complete history
in totally-ordered-time games.
The restriction imposed on any strategy $\sigma_i$ of any player $i$ consists of three components.
The first is traceability in \cite{Kamada2021}.
The second is that
in any subgame,
in any complete history consistent with strategy $\sigma_i$,
when times are partitioned into intervals during which
player $i$'s actions are constant,
any set of some of these intervals has the earliest interval.
The third is that
in any subgame,
any two complete histories consistent with strategy $\sigma_i$ coincide during a sufficiently small initial interval.
We show that
each tuple of strategies satisfying this restriction induces a unique complete history.

The restriction proposed by this paper is weaker than the restrictions in the existing literature.
If the inertiality in \cite{Bergin1993} is satisfied,
or
if the traceability and frictionality in \cite{Kamada2021} are satisfied,
the restriction in this paper is satisfied.
Both strategies satisfying the inertiality and strategies satisfying the traceability and frictionality
are natural
and
should be given each player.
Thus,
a restriction weaker than the both restriction is needed.
This paper
responds to this need
and
provides a generic restriction on strategies.

Totally-ordered-time games defined by this paper include continuous-time games and discrete-time games in the existing literature.
This paper's restriction does not restrict strategies in well-ordered time games,
which include discrete-time games.
Each strategy tuple induces a unique complete history with no restriction in discrete-time games,
and
thus,
restrictions imposed on strategies in totally-ordered-time games are required to degenerate into no restriction in the case of discrete-time games.
This paper's restriction satisfies this requirement.

This paper considers deterministic situations.
\cite{Bergin1993} considered deterministic situations,
whereas \cite{Kamada2021} considered stochastic situations.
Measurability of strategy tuples
does not matter in the former situations
but
matters in the latter situations.
This paper
focuses on deterministic situations
and
does not consider measurability of strategy tuples.

The remainder of this paper is organized as follows.
Section \ref{sec:model} describes the model.
Section \ref{sec:results} presents the results.
Section \ref{sec:literature} discusses relation to the literature.
Section \ref{sec:conclusion} concludes the paper.

\section{Model}\label{sec:model}

Let $\rnd{N,\rnd{T,\leq},\rnd{A_i}_{i \in N},\rnd{u_i}_{i \in N}}$ be a quadruple as follows.
\begin{itemize}
 \item
 $N$ is a nonempty finite set.
 $N$ represents the set of players.
 \item
 $\rnd{T,\leq}$ is a totally ordered set such that
 $T$ has a minimum,
 any nonempty $S \subset T$ has an infimum,
 and
 any nonempty $S \subset T$ bounded from above has a supremum.
 $T$ represents the set of times.
 \item
 For any $i \in N$,
 $A_i$ is a nonempty set.
 $A_i$ represents the set of player $i$'s actions.
 \item
 For any $i \in N$,
 $\map{u_i}{\rnd{\prod_{j \in N} A_j}^T}{\mathbb R}$.
 $\prod_{j \in N} A_i$,
 $\rnd{\prod_{j \in N} A_i}^T$
 and
 $u_i$
 represent
 the set of action tuples at each time,
 the set of complete histories
 and
 player $i$'s payoff function,
 respectively.
\end{itemize}

Introduce additional notations as follows.
\begin{itemize}
 \item
 Let $<$ ($\geq$; $>$, resp.)\ be the binary relation on $T$ such that
 for any $t,s \in T$,
 $\rnd{t < s} \leftrightarrow \rnd{t \leq s} \wedge \rnd{t \neq s}$
 ($\rnd{t \geq s} \leftrightarrow \rnd{s \leq t}$;
 $\rnd{t > s} \leftrightarrow \rnd{t \geq s} \wedge \rnd{t \neq s}$,
 resp.).
 \item
 For
 any $R \in \crly{\leq,<,\geq,>}$
 and
 any $t \in T$,
 let
 $T_{\mathrel{R} t} := \set{s \in T}{s \mathrel{R} t}$.
 \item
 For any $t \in T$,
 let $T^t := \set{s \in T}{s \geq t}$.
 \item
 For any $t,s \in T$,
 let
 $\rnd{t,s} := T_{> t} \cap T_{< s}$,
 $\rndsqr{t,s} := T_{> t} \cap T_{\leq s}$,
 $\sqrrnd{t,s} := T_{\geq t} \cap T_{< s}$
 and
 $\sqr{t,s} := T_{\geq t} \cap T_{\leq s}$.
 \item
 Let $H := \rnd{\prod_{i \in N} A_i}^T$.
 \item
 For
 any $h \in H$,
 any $i \in N$,
 any $S \subset T$
 and
 any $t \in T$,
 let
 \begin{itemize}
  \item
  $\map{h_i}{T}{A_i}$ such that
  for any $s \in T$,
  $h_i\rnd{s} = h\rnd{s}_i$,
  \item
  $\map{h_{-i}}{T}{\prod_{j \in N \setminus \crly{i}} A_j}$ such that
  for
  any $j \in N \setminus \crly{i}$
  and
  any $s \in T$,
  $h_{-i}\rnd{s}_j = h\rnd{s}_j$,
  \item
  $h^S$ ($h_i^S$; $h_{-i}^S$, resp.)\ be the restriction of $h$ ($h_i$; $h_{-i}$, resp.)\ to $S$
  and
  \item
  $h^t$ ($h_i^t$; $h_{-i}^t$, resp.)\ be $h^{T_{< t}}$ ($h_i^{T_{< t}}$; $h_{-i}^{T_{< t}}$, resp.).
 \end{itemize}
\end{itemize}

Introduce notations regarding strategies as follows.
\begin{itemize}
 \item
 Let
 $
  \Sigma_i
  :=
  \set
  {
   \sigma_i \in A_i^{T \times H}
  }
  {
   \rnd{\forall t \in T}
   \rnd{\forall h,g \in H}
   \rnd
   {
    \rnd{h^t = g^t}
    \rightarrow
    \rnd{\sigma_i\rnd{t,h} = \sigma_i\rnd{t,g}}
   }
  }
 $.
 $\Sigma_i$ represents the set of player $i$'s strategies.
 \item
 Let $\Sigma := \prod_{i \in N} \Sigma_i$.
 $\Sigma$ represents the set of strategy tuples.
 \item
 For
 any $i \in N$,
 any $\sigma_i \in \Sigma_i$,
 any $t \in T$
 and
 any $h \in H$,
 let $\sigma_i^t\rnd{h} := \sigma_i\rnd{t,h}$.
 \item
 For
 any $\sigma \in \Sigma$,
 any $t \in T$
 and
 any $h \in H$,
 let $\sigma^t\rnd{h} := \rnd{\sigma_i^t\rnd{h}}_{i \in N}$.
\end{itemize}

Introduce notations regarding feasibility as follows.
\begin{itemize}
 \item
 For any $i \in N$,
 $\map{\bar A_i}{T \times H}{2^{A_i} \setminus \crly{\emptyset}}$ such that
 for
 any $h,g \in H$
 and
 any $t \in T$,
 if $h^t = g^t$,
 $\bar A_i\rnd{t,h} = \bar A_i\rnd{t,g}$.
 \item
 $\bar A_i\rnd{t,h}$ represents the set of player $i$'s feasible actions at history $h^t$.
 \item
 For
 any $i \in N$,
 any $t \in T$
 and
 any $h \in H$,
 let $\bar A_i^t\rnd{h} := \bar A_i\rnd{t,h}$.
 \item
 Let
 $
  \bar H
  =
  \set
  {
   h \in H
  }
  {
   \rnd{\forall t \in T}
   \rnd{h\rnd{t} \in \prod_{i \in N} \bar A_i^t\rnd{h}}
  }
 $.
 \item
 $\bar H$ represents the set of feasible complete histories, i.e., complete histories $h$ such that
 for any period $t$,
 $h\rnd{t}$ is feasible action tuple at history $h^t$.
 \item
 For any $i \in N$,
 let
 $
  \bar\Sigma_i
  :=
  \set
  {
   \sigma_i \in \Sigma_i
  }
  {
   \rnd{\forall \rnd{t,h} \in T \times H} 
   \rnd{\sigma_i^t\rnd{h} \in \bar A_i^t\rnd{h}}
  }
 $.
 \item
 $\bar \Sigma_i$ represents the set of player $i$'s feasible strategies, i.e., player $i$'s strategies $\sigma_i$ such that
 at any history $h^t$,
 $\sigma_i^t\rnd{h}$ is a player $i$'s feasible action.
 \item
 Let
 $\bar\Sigma := \prod_{i \in N} \bar\Sigma_i$.
 $\bar \Sigma$ represents the set of feasible strategy tuples.
\end{itemize}
We do not explicitly consider feasibility.
However,
the result without imposing feasibility also holds with imposing feasibility.
That is,
as a main result of this paper,
it is shown that
if each player's strategies satisfy a set of axioms,
any strategy tuple induces a unique complete history;
by the same reasoning,
it is shown that
if each player's feasible strategies satisfy these axioms,
any feasible strategy tuple induces a unique feasible complete history.

\section{Results}\label{sec:results}

When
for any period $s \geq t$,
$h_i\rnd{s}$ coincides with the action specified by player $i$'s strategy $\sigma_i$ at history $h^s$,
we say that complete history $h$ is $t$-consistent with $\sigma_i$ for $i$.
When
for any $i \in N$,
$h$ is $t$-consistent with $\sigma_i$ for $i$,
we say that complete history $h$ is $t$-consistent with $\sigma$.
Definition \ref{def:consistency} is owing to \cite{Kamada2021}.
\begin{definition}\label{def:consistency}
 Let
 $t \in T$,
 $i \in N$,
 $h \in H$
 and
 $\sigma_i \in \Sigma_i$.
 $h$ is \emph{$t$-consistent with $\sigma_i$ for $i$}
 if and only if
 for any $s \in T_{\geq t}$
 $h_i\rnd{s} = \sigma_i^s\rnd{h}$.
 For
 any $t \in T$,
 any $h \in H$
 and
 any $\sigma \in \Sigma$,
 $h$ is \emph{$t$-consistent with $\sigma$}
 if and only if
 for any $i \in N$,
 $h$ is $t$-consistent with $\sigma_i$ for $i$.
\end{definition}

For
any $t \in T$,
any $i \in N$,
any $\sigma_i \in \Sigma_i$
and
any $\sigma \in \Sigma$,
let
$\ch_i^t\rnd{\sigma_i}$ be the set of $h \in H$ that is $t$-consistent with $\sigma_i$ for $i$
and
$\ch^t\rnd{\sigma}$ be the set of $h \in H$ that is $t$-consistent with $\sigma$.
For
any $t \in T$,
any $i \in N$
and
any $h \in H$,
let
$
 \sh^t\rnd{h}
 :=
 \set
 {
  g \in H
 }
 {
  g^t = h^t
 }
$
and
$
 \sh_i^t\rnd{h}
 :=
 \set
 {
  g \in \sh^t\rnd{h}
 }
 {
  g_{-i} = h_{-i}
 }
$.

Axiom \ref{ax:traceability} states that
there exists $g \in \sh_i^t\rnd{h}$ such that
$g$ is $t$-consistent with strategy $\sigma_i$ for player $i$.
Axiom \ref{ax:traceability} is owing to \cite{Kamada2021}.
\begin{axiom}\label{ax:traceability}
 Let
 $i \in N$
 and
 $\sigma_i \in \Sigma_i$.
 For
 any $t \in T$
 and
 any $h \in H$,
 $\sh_i^t\rnd{h} \cap \ch_i^t\rnd{\sigma_i} \neq \emptyset$.
\end{axiom}

Let $\leq$ be the partial order on $2^T \setminus \crly{\emptyset}$ such that
for any $S,R \in 2^T \setminus \crly{\emptyset}$,
$S \leq R$
if and only if
$
 \rnd{\forall s \in S}
 \rnd{\forall r \in R}
 \rnd{s < r}
$,
or $S = R$.
Let $<$ be the irreflexive part of $\leq$.
For any $\mathcal S \subset 2^T \setminus \crly{\emptyset}$,
$\rnd{\mathcal S,\leq}$ is a partially ordered set.
For any $t \in T$,
let $\mathcal C^t$ be the set of connected sets in $T^t$ equipped with the order topology on $\rnd{T^t,\leq}$.
For
any $i \in N$,
$t \in T$
and
$h \in H$,
let
\begin{align*}
 \pi_i^t\rnd{h}
 :=
 \set
 {
  S \in \mathcal C^t \setminus \crly{\emptyset}
 }
 {
  \rnd{\forall R \in \mathcal C^t}
  \rnd
  {
   \rnd{R \supset S}
   \rightarrow
   \rnd
   {
    \rnd{R = S}
    \leftrightarrow
    \rnd{\card{h_i\rnd{R}} = 1}
   }
  }
 }.
\end{align*}

Lemma \ref{lem:totally_ordered_partition} states that
$\pi_i^t\rnd{h}$ is a partition of $T^t$ that is totally ordered by $\leq$.
\begin{lemma}\label{lem:totally_ordered_partition}
 Let
 $t \in T$
 and
 $h \in H$.
 Then,
 $\pi_i^t\rnd{h}$ is a partition of $T^t$,
 and
 $\rnd{\pi_i^t\rnd{h},\leq}$ is a totally ordered set.
\end{lemma}
\begin{proof}
 See Section \ref{sec:proof_of_totally_ordered_partition}.
\end{proof}

Axiom \ref{ax:well-orderedness} states that
in
any subgame starting from period $t$
and
any complete history $t$-consistent with $\sigma_i$ for $i$,
when times are partitioned into connected sets (intervals) in which player $i$'s actions are constant,
any set consisting of some of these connected sets has a minimum set.
Axiom \ref{ax:well-orderedness} is introduced by this paper.
\begin{axiom}\label{ax:well-orderedness}
 Let
 $i \in N$
 and
 $\sigma_i \in \Sigma_i$.
 For
 any $t \in T$
 and
 any $h \in \ch_i^t\rnd{\sigma_i}$,
 $\rnd{\pi_i^t\rnd{h},\leq}$ is a well-ordered set.
\end{axiom}

Axiom \ref{ax:initial_uniqueness} states that
in any subgame starting from period $t$,
any two complete histories $t$-consistent with strategy $\sigma_i$ for player $i$ coincide during a sufficiently small initial interval.
Axiom \ref{ax:initial_uniqueness} is introduced by this paper.
\begin{axiom}\label{ax:initial_uniqueness}
 Let
 $i \in N$
 and
 $\sigma_i \in \Sigma_i$.
 For
 any $t \in T$ such that $T_{> t} \neq \emptyset$,
 and
 any $h,g \in \ch_i^t\rnd{\sigma_i}$ such that
 $h^t = g^t$,
 there exists $s \in T_{> t}$ such that
 $h_i^{\sqrrnd{t,s}} = g_i^{\sqrrnd{t,s}}$.
\end{axiom}

For any $i \in N$,
let $\tilde\Sigma_i$ be the set of $\sigma_i \in \Sigma_i$ that satisfies Axioms \ref{ax:traceability}--\ref{ax:initial_uniqueness}.
Let $\tilde\Sigma := \prod_{i \in N} \tilde\Sigma_i$.

Theorem \ref{thm:unique_existence_of_consistent_history} states that
any strategy tuple such that each player's strategy satisfies Axioms \ref{ax:traceability}--\ref{ax:initial_uniqueness} has a unique consistent history in any subgame.
\begin{theorem}\label{thm:unique_existence_of_consistent_history}
 Let
 $t \in T$,
 $h \in H$
 and
 $\sigma \in \tilde\Sigma$.
 Then,
 $\sh^t\rnd{h} \cap \ch^t\rnd{\sigma}$ is a singleton.
\end{theorem}
\begin{proof}
 See Appendix \ref{sec:proof_of_unique_existence_of_consistent_history}.
\end{proof}

\section{Relationship with literature}\label{sec:literature}

Proposition \ref{prop:well-ordered_time} states that
if times are well-ordered,
any strategy of any player satisfies Axioms \ref{ax:traceability}--\ref{ax:initial_uniqueness}.
\begin{proposition}\label{prop:well-ordered_time}
 Suppose that
 $\rnd{T,\leq}$ is a well-ordered set.
 Let
 $i \in N$
 and
 $\sigma_i \in \Sigma_i$.
 Then,
 $\sigma_i$ satisfies Axioms \ref{ax:traceability}--\ref{ax:initial_uniqueness}.
\end{proposition}
\begin{proof}
 See Appendix \ref{sec:proof_of_well-ordered_time}.
\end{proof}

Axiom \ref{ax:inertiality} states that
in any subgame starting from period $t$,
there exists a small initial interval during which player $i$ does not change his/her action in any complete history $t$-consistent with strategy $\sigma_i$ for $i$.
Axiom \ref{ax:inertiality} is owing to \cite{Bergin1993}.
\begin{axiom}\label{ax:inertiality}
 Let
 $i \in N$
 and
 $\sigma_i \in \Sigma_i$.
 For
 any $t \in T$ such that $T_{> t} \neq \emptyset$
 and
 any $h \in H$,
 there exist
 $s \in T_{> t}$
 and
 $a_i \in A_i$
 such that
 for
 any $r \in \sqrrnd{t,s}$
 and
 any $g \in \sh^t\rnd{h}$,
 $\sigma_i^r\rnd{g} = a_i$.
\end{axiom}

Proposition {prop:inertiality} states that
Axiom \ref{ax:inertiality} implies Axioms \ref{ax:traceability}--\ref{ax:initial_uniqueness}.
\begin{proposition}\label{prop:inertiality}
 Let
 $i \in N$
 and
 $\sigma_i \in \Sigma_i$.
 Then,
 if $\sigma_i$ satisfies Axiom \ref{ax:inertiality},
 it satisfies Axioms \ref{ax:traceability}--\ref{ax:initial_uniqueness}.
\end{proposition}
\begin{proof}
 See Appendix \ref{sec:proof_of_inertiality}.
\end{proof}

Axiom states that
in any subgame starting from period $t$,
player $i$ moves only finite times during any finite-length interval in any complete history $t$-consistent with strategy $\sigma_i$ for $i$.
Axiom \ref{ax:frictionality} is owing to \cite{Kamada2021}.
\begin{axiom}\label{ax:frictionality}
 Let
 $i \in N$
 and
 $\sigma_i \in \Sigma_i$.
 There exists $z_i \in A_i$ such that
 for
 any $t \in T$,
 any $h \in \ch_i^t\rnd{\sigma_i}$
 and
 any $s \in T_{> t}$,
 $
  \card
  {
   \set
   {q \in \sqr{t,s}}
   {h_i\rnd{q} \neq z_i}
  }
  <
  \infty
 $.
\end{axiom}

Proposition \ref{prop:frictionality} states that
Axiom \ref{ax:frictionality} implies Axioms \ref{ax:well-orderedness} and \ref{ax:initial_uniqueness}.
\begin{proposition}\label{prop:frictionality}
 Let
 $i \in N$
 and
 $\sigma_i \in \Sigma_i$.
 Then,
 if $\sigma_i$ satisfies Axiom \ref{ax:frictionality},
 it satisfies Axioms \ref{ax:well-orderedness} and \ref{ax:initial_uniqueness}.
\end{proposition}
\begin{proof}
 See Appendix \ref{sec:proof_of_frictionality}.
\end{proof}

\section{Conclusion}\label{sec:conclusion}

We defined deterministic totally-ordered-time games.
We show that
for any tuple of strategies that satisfy the three axioms,
in any subgame,
there exists a unique complete history that is consistent with the strategy tuple.

It is a open question whether or not
there is a natural weaker set of axioms of strategies that makes any strategy tuple induces a unique complete history.

\newpage
\appendix
\section*{Appendix}

For any $t \in T$,
let
$
 \mathcal T^t
 :=
 \set
 {S \in \mathcal C^t}
 {S \ni t}
$.
Let $\mathcal T := \bigcup_{t \in T} \mathcal T^t$.
For
any $S \in 2^T$
and
any $t \in T$,
let
\begin{itemize}
 \item
 $\Pi^S$ be the set of partitions of $S$, 
 \item
 $\hat\Pi^S$ be the set of $\pi \in \Pi^S$ such that
 $\rnd{\pi,\leq}$ is a totally ordered set,
 \item
 $\hat{\hat\Pi}^S$ be the set of $\pi \in \hat\Pi^S$ such that
 $\rnd{\pi,\leq}$ is a well-ordered set,
 \item
 $\Pi^t := \Pi^{T_{\geq t}}$,
 \item
 $\hat\Pi^t := \hat\Pi^{T_{\geq t}}$
 and
 \item
 $\hat{\hat\Pi}^t := \hat{\hat\Pi}^{T_{\geq t}}$.
\end{itemize}
For any set $\mathcal A$,
let $C_{\mathcal A}$ be the set of $\map{f}{\mathcal A}\rnd{\bigcup \mathcal A}$
such that
for any $A \in \mathcal A$,
$f\rnd{A} \in A$.
For
any $t \in T$,
any $\pi,\rho \in \Pi^t$
and
any $\Upsilon \subset \Pi^t$,
let
\begin{itemize}
 \item
 $
  \pi \cap \rho
  :=
  \set
  {
   S \cap R
  }
  {
   \rnd{\rnd{S,R} \in \pi \times \rho}
   \wedge
   \rnd{S \cap R \neq \emptyset}
  }
 $
 and
 \item
 $
  \bigcap \Upsilon
  :=
  \set
  {
   \bigcap S\rnd{\Upsilon}
  }
  {
   \rnd{S \in C_\Upsilon}
   \wedge
   \rnd{\bigcap S\rnd{\Upsilon} \neq \emptyset}
  }
 $.\footnote
 {
  $S \in C_\Upsilon$ assigns each partition $\pi \in \Upsilon$ with a set $S\rnd{\pi} \in \pi$,
  and
  $\bigcap S\rnd{\Upsilon} = \bigcup_{\pi \in \Upsilon} S\rnd{\pi}$.
 }
\end{itemize}
For
any $i \in N$,
any $S \in 2^T$,
any $\sigma_i \in \Sigma_i$
and
any $\sigma \in \Sigma$,
let
\begin{itemize}
 \item
 $
  \ch_i^S\rnd{\sigma_i}
  := 
  \set
  {
   h \in H
  }
  {
   \rnd{\forall t \in S}
   \rnd{h_i\rnd{t} = \sigma_i^t\rnd{g}}
  }
 $
 and
 \item
 $
  \ch^S\rnd{\sigma}
  := 
  \set
  {
   h \in H
  }
  {
   \rnd{\forall t \in S}
   \rnd{h\rnd{t} = \sigma^t\rnd{g}}
  }
 $.
\end{itemize}

\section{Lemmas}\label{sec:lemmas}

\begin{lemma}\label{lem:increment_of_uniqueness}
 Let
 $i \in N$,
 $\sigma_i \in \Sigma_i$,
 $t \in T$,
 $s \in T_{\geq t}$,
 $S \in \mathcal T^t$
 and
 $h,g \in \ch_i^{\sqr{t,s} \cap S}\rnd{\sigma_i}$ such that
 $h^t = g^t$.
 Suppose that
 $h^{\sqrrnd{t,s} \cap S} = g^{\sqrrnd{t,s} \cap S}$.
 Then,
 $h_i^{\sqr{t,s} \cap S} = g_i^{\sqr{t,s} \cap S}$.
\end{lemma}
\begin{proof}
 Consider the case where
 $s \in S$.
 Then,
 $
  \sqrrnd{t,s}
  \subset
  \sqr{t,s}
  \subset
  S
 $.
 Thus,
 by the supposition,
 $h^s = g^s$.
 Hence,
 $
  h_i\rnd{s} 
  =
  \sigma_i^s\rnd{h}
  =
  \sigma_i^s\rnd{g}
  =
  g_i\rnd{s}
 $.
 Thus,
 $h_i^{\sqr{t,s}} = g_i^{\sqr{t,s}}$.
 Note that
 $\sqr{t,s} \subset S$.
 Then,
 $h_i^{\sqr{t,s} \cap S} = g_i^{\sqr{t,s} \cap S}$.

 Consider the case where
 $s \notin S$.
 Then,
 because $S \in \mathcal T^t$,
 for any $r \in S$,
 $r < s$.
 Thus,
 $
  S
  \subset
  \sqrrnd{t,s}
  \subset
  \sqr{t,s}
 $.
 Thus,
 by the supposition,
 $h_i^{\sqr{t,s} \cap S} = g_i^{\sqr{t,s} \cap S}$.
\end{proof}

\begin{lemma}\label{lem:increment_of_uniqueness'}
 Let
 $\sigma \in \Sigma$,
 $t \in T$,
 $s \in T_{\geq t}$,
 $S \in \mathcal T^t$
 and
 $h,g \in \ch^{\sqr{t,s} \cap S}\rnd{\sigma}$ such that
 $h^t = g^t$.
 Suppose that
 $h^{\sqrrnd{t,s} \cap S} = g^{\sqrrnd{t,s} \cap S}$.
 Then,
 $h^{\sqr{t,s} \cap S} = g^{\sqr{t,s} \cap S}$.
\end{lemma}
\begin{proof}
 By Lemma \ref{lem:increment_of_uniqueness},
 for any $i \in N$,
 $h_i^{\sqr{t,s} \cap S} = g_i^{\sqr{t,s} \cap S}$.
 Thus,
 $h^{\sqr{t,s} \cap S} = g^{\sqr{t,s} \cap S}$.
\end{proof}

\begin{lemma}\label{lem:increment_of_existence}
 Let
 $i \in N$,
 $\sigma_i \in \Sigma_i$,
 $t \in T$,
 $h \in H$
 and
 $s \in T_{\geq t}$.
 Suppose that
 $\sh_i^t\rnd{h} \cap \ch_i^{\sqrrnd{t,s}}\rnd{\sigma_i} \neq \emptyset$.
 Then,
 $\sh_i^t\rnd{h} \cap \ch_i^{\sqr{t,s}}\rnd{\sigma_i} \neq \emptyset$.
\end{lemma}
\begin{proof}
 By the supposition,
 there exists $g \in \sh_i^t\rnd{h} \cap \ch_i^{\sqrrnd{t,s}}\rnd{\sigma_i}$.
 There exists $f \in \sh_i^s\rnd{g}$ such that
 $f_i\rnd{s} = \sigma_i^s\rnd{g}$.
 For any $r \in \sqrrnd{t,s}$,
 because
 $f^s = g^s$,
 $g \in \ch_i^{\sqrrnd{t,s}}\rnd{\sigma_i}$,
 and
 $f^r = g^r$,
 $
  f_i\rnd{r}
  =
  g_i\rnd{r}
  =
  \sigma_i^r\rnd{g}
  =
  \sigma_i^r\rnd{f}
 $.
 Because
 $f_i\rnd{s} = \sigma_i^s\rnd{g}$,
 and
 $f^s = g^s$,
 $f_i\rnd{s} = \sigma_i^s\rnd{f}$.
 Thus,
 $f \in \ch_i^{\sqr{t,s}}\rnd{\sigma_i}$.
 Note that
 because
 $f \in \sh_i^s\rnd{g}$,
 and
 $t < s$,
 $f \in \sh_i^t\rnd{g}$,
 and
 thus,
 because $g \in \sh_i^t\rnd{h}$,
 $f \in \sh_i^t\rnd{h}$.
 Then,
 $f \in \sh_i^t\rnd{h} \cap \ch_i^{\sqr{t,s}}\rnd{\sigma_i}$.
 Thus,
 $\sh_i^t\rnd{h} \cap \ch_i^{\sqr{t,s}}\rnd{\sigma_i} \neq \emptyset$.
\end{proof}

\begin{lemma}\label{lem:increment_of_existence'}
 Let
 $\sigma \in \Sigma$,
 $t \in T$,
 $h \in H$
 and
 $s \in T_{\geq t}$.
 Suppose that
 $\sh^t\rnd{h} \cap \ch^{\sqrrnd{t,s}}\rnd{\sigma} \neq \emptyset$.
 Then,
 $\sh^t\rnd{h} \cap \ch^{\sqr{t,s}}\rnd{\sigma} \neq \emptyset$.
\end{lemma}
\begin{proof}
 By the supposition,
 there exists $g \in \sh^t\rnd{h} \cap \ch^{\sqrrnd{t,s}}\rnd{\sigma}$.
 There exists $f \in \sh^s\rnd{g}$ such that
 $f\rnd{s} = \sigma^s\rnd{g}$.
 For any $r \in \sqrrnd{t,s}$,
 because
 $f^s = g^s$,
 $g \in \ch^{\sqrrnd{t,s}}\rnd{\sigma}$,
 and
 $f^r = g^r$,
 $
  f\rnd{r}
  =
  g\rnd{r}
  =
  \sigma^r\rnd{g}
  =
  \sigma^r\rnd{f}
 $.
 Because
 $f\rnd{s} = \sigma^s\rnd{g}$,
 and
 $f^s = g^s$,
 $f\rnd{s} = \sigma^s\rnd{f}$.
 Thus,
 $f \in \ch^{\sqr{t,s}}\rnd{\sigma}$.
 Note that
 because
 $f \in \sh^s\rnd{g}$,
 and
 $t < s$,
 $f \in \sh^t\rnd{g}$,
 and
 thus,
 because $g \in \sh^t\rnd{h}$,
 $f \in \sh^t\rnd{h}$.
 Then,
 $f \in \sh^t\rnd{h} \cap \ch^{\sqr{t,s}}\rnd{\sigma}$.
 Thus,
 $\sh^t\rnd{h} \cap \ch^{\sqr{t,s}}\rnd{\sigma} \neq \emptyset$.
\end{proof}

\section{Proof of Lemma \ref{lem:totally_ordered_partition}}\label{sec:proof_of_totally_ordered_partition}

\begin{lemma}\label{lem:partition}
 $\pi_i^t\rnd{h} \in \Pi^t$.
\end{lemma}
\begin{proof}
 Show that
 $\pi_i^t\rnd{h}$ is a cover of $T^t$.
 Let $s \in T_{\geq t}$.
 Let
 \begin{align*}
  \mathcal S
  :=
  \set
  {
   R \in \mathcal C^t
  }
  {
   \rnd{R \ni s}
   \wedge
   \rnd{h_i\rnd{R} = \crly{h_i\rnd{s}}}
  }
 \end{align*}
 and
 $S := \bigcup \mathcal S$.
 Because
 $\crly{s} \in \mathcal C^t$,
 $\crly{s} \ni s$,
 and
 $h_i\rnd{\crly{s}} = \crly{h_i\rnd{s}}$,
 $\crly{s} \in \mathcal S$;
 thus,
 $\crly{s} \subset S$;
 hence,
 $s \in S$.
 By the following,
 $S \in \pi_i^t\rnd{h}$.
 \begin{itemize}
  \item
  Because the union of any connected sets that have a common member is connected,
  $S \in \mathcal C^t$.
  \item
  Because $S \ni s$,
  $S \neq \emptyset$.
  \item
  For any $r \in S$,
  by the definition of $S$,
  there exists $R \in \mathcal S$ such that
  $R \ni r$;
  $h_i\rnd{R} = \crly{h_i\rnd{s}}$;
  thus,
  $h_i\rnd{r} = h_i\rnd{s}$.
  Hence,
  $\card{h_i\rnd{S}} = 1$.
  \item
  Let $R \in \mathcal C^t$ such that
  $R \supsetneq S$.
  Then,
  $R \notin \mathcal S$.
  Note that
  $R \in \mathcal C^t$,
  and
  $R \supsetneq S \ni s$.
  Then,
  $h_i\rnd{R} \neq \crly{h_i\rnd{s}}$.
  Note that
  $R \ni s$,
  and
  thus,
  $h_i\rnd{R} \supset \crly{h_i\rnd{s}}$.
  Then,
  $\card{R} > 1$.
 \end{itemize}
 Thus,
 for any $s \in T^t$,
 there exists $S \in \pi_i^t\rnd{h}$ such that
 $S \ni s$.

 Show that
 for any two sets in $\pi_i^t\rnd{h}$,
 if they have a common member,
 they are the same set.
 Let $S,R \in \pi_i^t\rnd{h}$.
 Suppose that
 there exists $s \in S \cap R$.
 Because
 $S,R \in \mathcal C^t$,
 and
 $s \in S \cap R$,
 $S \cup R \in \mathcal C^t$.
 For any $Q \in \crly{S,R}$,
 because
 $\card{h_i\rnd{Q}} = 1$,
 and
 $Q \ni s$,
 $h_i\rnd{Q} = \crly{h_i\rnd{s}}$;
 thus,
 $h_i\rnd{S \cup R} = \crly{h_i\rnd{s}}$;
 hence,
 $\card{h_i\rnd{S \cup R}} = 1$.
 Thus,
 for any $Q \in \crly{S,R}$
 by the definition of $\pi_i^t\rnd{h}$ and $Q \in \pi_i^t\rnd{h}$,
 $S \cup R = Q$.
 Thus,
 $S = R$.

 Hence,
 $\pi_i^t\rnd{h} \in \Pi^t$.
\end{proof}

\begin{lemma}\label{lem:totally_orderedness}
 $\pi_i^t\rnd{h} \in \hat\Pi^t$.
\end{lemma}
\begin{proof}
 Let $S,R \in \pi_i^t\rnd{h}$.
 Suppose that
 $S \leq R$,
 and
 $R \leq S$.
 Suppose that
 $S \neq R$
 (assumption for contradiction).
 Then,
 $S < R$,
 and
 $R < S$.
 There exist $s \in S$ and $r\in R$.
 Because
 $S < R$,
 and
 $R < S$,
 $s < r$,
 and
 $r < s$,
 which is a contradiction.
 Thus,
 $S = R$.
 Hence,
 $\leq$ satisfies antisymmetry.

 Let $S,R,Q \in \pi_i^t\rnd{h}$.
 Suppose that
 $S \leq R$,
 and
 $R \leq Q$.
 Consider the case where
 $S = R$,
 or
 $R = Q$.
 Then,
 $S \leq Q$.
 Consider the case
 $S \neq R$,
 and
 $R \neq Q$.
 Then,
 $S < R$,
 and
 $R < Q$.
 Let
 $s \in S$
 and
 $q \in Q$.
 There exists $r \in R$.
 Because
 $S < R$,
 and
 $R < Q$,
 $s < r$,
 and
 $r < q$.
 Thus,
 $s < q$.
 Hence,
 $S < Q$.
 Thus,
 $\leq$ satisfies transitivity.

 Let $S,R \in \pi_i^t\rnd{h}$.
 Suppose that
 $
  \neg
  \rnd
  {
   \rnd{S \leq R}
   \vee
   \rnd{R \leq S}
  }
 $
 (assumption for contradiction).
 Note that
 \begin{align*}
  \neg
  \rnd
  {
   \rnd{S \leq R}
   \vee
   \rnd{R \leq S}
  }&
  \leftrightarrow
  \neg
  \rnd{S \leq R}
  \wedge
  \neg
  \rnd{R \leq S}\\&
  \leftrightarrow
  \neg
  \rnd
  {
   S < R
   \vee
   S = R
  }
  \wedge
  \neg
  \rnd
  {
   R < S
   \vee
   R = S
  }\\&
  \leftrightarrow
  \neg
  \rnd{S < R}
  \wedge
  \neg
  \rnd{R < S}
  \wedge
  \rnd{S \neq R}.
 \end{align*}
 Then,
 there exist $l_S,u_S \in S$ and $l_R,u_R \in R$ such that
 $
  \neg
  \rnd{u_S < l_R}
  \wedge
  \neg
  \rnd{u_R < l_S}
 $;
 by Lemma \ref{lem:partition},
 $l_S,u_S \notin R$,
 and
 $l_R,u_R \notin S$.
 Consider the case where
 $l_S \leq l_R$.
 Then,
 because $\neg \rnd{u_S < l_R}$,
 $l_S \leq l_R \leq u_S$.
 Note that
 $S$ is a connected set in $T^t$ equipped with the order topology on $\rnd{T^t,\leq}$,
 thus,
 $S$ is an interval in $\rnd{T^t,\leq}$,
 and
 hence,
 $\sqr{l_S,u_S} \subset S$.
 Then,
 $l_R \in S$,
 which contradicts that
 $l_R \notin S$.
 Similarly,
 in the case where
 $l_R \leq l_S$,
 $l_S \in R$,
 which contradicts that
 $l_S \notin R$.
 Thus,
 $
  \rnd{S \leq R}
  \vee
  \rnd{R \leq S}
 $.
 Hence,
 $\leq$ satisfies totality.

 Thus,
 $\rnd{\pi_i^t\rnd{h},\leq}$ is a totally ordered set.
 Hence,
 by Lemma \ref{lem:partition},
 $\pi_i^t\rnd{h} \in \hat\Pi^t$.
\end{proof}

The conclusion follows from Lemma \ref{lem:totally_orderedness}.
\qed

\section{Proof of Theorem \ref{thm:unique_existence_of_consistent_history}}\label{sec:proof_of_unique_existence_of_consistent_history}

Let $t \in T$.
Then,
Lemmas \ref{lem:intersection_of_two_partitions} and \ref{lem:intersection_of_finite_partitions} are shown.

\begin{lemma}\label{lem:intersection_of_two_partitions}
 Let $\pi,\rho \in \hat{\hat\Pi}^t$.
 Then,
 $\pi \cap \rho \in \hat{\hat\Pi}^t$.
\end{lemma}
\begin{proof}
 By the following,
 $\pi \cap \rho \in \Pi^t$.
 \begin{itemize}
  \item
  Let $s \in T^t$.
  There exists $S \in \pi$ and $R \in \rho$ such that
  $S \in s$,
  and
  $R \in s$.
  $S \cap R \in \pi \cap \rho$.
  $s \in S \cap R$.
  Thus,
  $\pi \cap \rho$ is a cover of $T^t$.
  \item
  Let $S,R \in \pi \cap \rho$.
  Suppose that
  there exists $s \in S \cap R$.
  For any $Q \in \crly{S,R}$,
  there exist
  $Q_\pi \in \pi$ and $Q_\rho \in \rho$ such that
  $Q_\pi \cap Q_\rho = Q$.
  For any $\tau \in \crly{\pi,\rho}$,
  because
  $S_\tau,R_\tau \in \tau$,
  and
  $S_\tau,R_\tau \ni s$,
  $S_\tau = R_\tau$.
  Thus,
  $S_\pi \cap S_\rho = R_\pi \cap R_\rho$.
  Hence,
  $S = R$.
 \end{itemize}

 For any $S \in \pi \cap \rho$,
 there exist $R \in \pi$ and $Q \in \rho$ such that
 $R \cap Q = S$,
 and
 $R \cap Q \neq \emptyset$;
 because
 $R,Q \in \mathcal C^t$,
 and
 $R \cap Q \neq \emptyset$,
 $S = R \cap Q \in \mathcal C^t$;
 because $S$ is a connected set in $T^t$ equipped with the order topology on $\rnd{T^t,\leq}$,
 $S$ is an interval in $\rnd{T^t,\leq}$.
 Thus,
 $\rnd{\pi \cap \rho,\leq}$ is a totally ordered set.
 Hence,
 $\pi \cap \rho \in \hat\Pi^t$.

 Let $\tau \in 2^{\pi \cap \rho} \setminus \crly{\emptyset}$.
 For any $\upsilon \in \crly{\pi,\rho}$,
 let
 $
  \upsilon'
  :=
  \set
  {
   S \in \upsilon
  }
  {
   \rnd{\exists R \in \xi}
   \rnd{R \cap S \in \tau}
  }
 $,
 where $\xi \in \crly{\pi,\rho} \setminus \crly{\upsilon}$.
 For any $\upsilon \in \crly{\pi,\rho}$,
 $\upsilon' \neq \emptyset$.
 Because $\pi,\rho \in \hat{\hat\Pi}$,
 for any $\upsilon \in \crly{\pi,\rho}$,
 there exists a minimum $S_\upsilon$ of $\upsilon'$.
 Let $S := S_\pi \cap S_\rho$.
 Suppose that
 $S \notin \tau$
 (assumption for contradiction).
 Then,
 for any $\upsilon \in \crly{\pi,\rho}$,
 there exists $R_\upsilon \in \upsilon \setminus \crly{S_\upsilon}$ such that
 $R_\upsilon \cap S_\xi \in \tau$,
 and
 thus,
 there exists $s_\upsilon \in R_\upsilon \cap S_\xi$.
 where $\xi \in \crly{\pi,\rho} \setminus \crly{\upsilon}$.
 For any $\upsilon$,
 by the definition of $\upsilon'$,
 $R_\upsilon \in \upsilon'$,
 and
 thus,
 because $S_\upsilon$ is a minimum of $\upsilon'$,
 $S_\upsilon < R_\upsilon$.
 Thus,
 for any $\upsilon$,
 because
 $s_\xi \in S_\upsilon$,
 and
 $s_\upsilon \in R_\upsilon$,
 $s_\xi < s_\upsilon$,
 where $\xi \in \crly{\pi,\rho} \setminus \crly{\upsilon}$.
 Hence,
 $s_\pi < s_\rho$,
 and
 $s_\rho < s_\pi$,
 which is a contradiction.
 Thus,
 $S \in \tau$.
 Let $Q \in \tau \setminus \crly{S}$.
 Then,
 there exist $Q_\pi \in \pi$ and $Q_\rho \in \rho$ such that
 $Q_\pi \cap Q_\rho = Q$.
 Because $Q \neq S$,
 there exists $\upsilon \in \crly{\pi,\rho}$ such that
 $Q_\upsilon \neq S_\upsilon$.
 By the definition of $\upsilon'$,
 $Q_\upsilon \in \upsilon'$.
 Thus,
 because $S_\upsilon$ is a minimum of $\upsilon'$,
 $S_\upsilon < Q_\upsilon$.
 Thus,
 $
  S
  =
  S_\pi \cap S_\rho
  <
  Q_\pi \cap Q_\rho
  =
  Q
 $.
 Hence,
 there exists a minimum of $\tau$.
 Thus,
 $\rnd{\pi \cap \rho,\leq}$ is a well-ordered set.
 Hence,
 $\pi \cap \rho \in \hat{\hat\Pi}^t$.
\end{proof}

\begin{lemma}\label{lem:intersection_of_finite_partitions}
 Let $\Upsilon \subset \hat{\hat\Pi}^t$ such that
 $\Upsilon$ is nonempty and finite.
 Then,
 $\bigcap \Upsilon \in \hat{\hat\Pi}^t$.
\end{lemma}
\begin{proof}
 It suffices to show that
 for any $n \in \mathbb N$ such that $n \leq \card{\Upsilon}$,
 for any $\Xi \subset \Upsilon$ such that $\card{\Xi} = n$,
 $\bigcap \Xi \in \hat{\hat\Pi}^t$.
 By mathematical induction,
 show it.

 For any $\Xi$ such that $\card{\Xi} = 1$,
 $\bigcap \Xi = \pi \in \hat{\hat\Pi}^t$,
 where $\pi$ is the unique member in $\Xi$.

 Let $n \in \mathbb N$ such that $2 \leq n \leq \card{\Upsilon}$.
 Suppose that
 for any $\Xi \subset \Upsilon$ such that
 $\card{\Xi} = n - 1$,
 $\bigcap \Xi \in \hat{\hat\Pi}^t$
 (induction hypothesis).
 Let $\Xi \in \Upsilon$ such that $\card{\Xi} = n$.
 There exists
 $\pi \in \Xi$.
 Note that
 by the induction hypothesis,
 $\bigcap \rnd{\Xi \setminus \crly{\pi}} \in \Pi^t$.
 For any set $S$,
 \begin{align*}
  \rnd{S \in \bigcap \Xi}&
  \leftrightarrow
  \rnd{\exists \tilde S \in C_\Xi}
  \rnd{
   S
   =
   \bigcap \tilde S\rnd{\Xi}
   \neq
   \emptyset
  }\\&
  \leftrightarrow
  \rnd{\exists \tilde R \in C_{\Xi \setminus \crly{\pi}}}
  \rnd{\exists Q \in \pi}
  \rnd
  {
   S
   =
   \rnd{\bigcap \tilde R\rnd{\Xi \setminus \crly{\pi}}} \cap Q
   \neq
   \emptyset
  }\\&
  \leftrightarrow
  \rnd{\exists R \in\bigcap \rnd{\Xi \setminus \crly{\pi}}}
  \rnd{\exists Q \in \pi}
  \rnd
  {
   S
   =
   R \cap Q
   \neq \emptyset
  }\\&
  \leftrightarrow
  \rnd{S \in \rnd{\bigcap \rnd{\Xi \setminus \crly{\pi}}} \cap \pi}.
 \end{align*}
 Thus,
 $\bigcap \Xi = \rnd{\bigcap \rnd{\Xi \setminus \crly{\pi}}} \cap \pi$.
 Hence,
 by the induction hypothesis and Lemma \ref{lem:intersection_of_two_partitions},
 $\bigcap \Xi \in \hat{\hat\Pi}^t$.
\end{proof}

Let
$i \in N$,
$\sigma_i \in \tilde\Sigma_i$,
$t \in T$ such that $T_{> t} \neq \emptyset$
and
$S \in \mathcal T^t$.
Then,
Lemmas \ref{lem:extension_of_consistency} and \ref{lem:initial_uniqueness} are shown.

\begin{lemma}\label{lem:extension_of_consistency}
 Let $h \in \ch_i^S\rnd{\sigma_i}$.
 Then,
 there exists $g \in \ch_i^t\rnd{\sigma_i}$ such that
 $g^t = h^t$,
 and
 $g_i^S = h_i^S$.
\end{lemma}
\begin{proof}
 Consider the case where
 $S$ is not bounded from above.
 Then,
 $S = T_{\geq t}$.
 Thus,
 $h \in \ch_i^t\rnd{\sigma_i}$.

 Consider the case where
 $S$ is bounded from above.
 Let $s := \sup S$.
 By Axiom \ref{ax:traceability},
 there exists
 $g \in \sh_i^s\rnd{h} \cap \ch_i^s\rnd{\sigma_i}$.
 By the following,
 $g \in \ch_i^t\rnd{\sigma_i}$,
 $g^t = h^t$,
 and
 $g_i^S = h_i^S$.
 \begin{itemize}
  \item
  For any $r \in S \setminus \crly{s}$,
  because
  $h \in \ch_i^S\rnd{\sigma_i}$,
  and
  $g \in \sh_i^s\rnd{h}$,
  $
   g_i\rnd{r}
   =
   h_i\rnd{r}
   =
   \sigma_i^r\rnd{h}
   =
   \sigma_i^r\rnd{g}
  $.
  Note that
  $g \in \ch_i^s\rnd{\sigma_i}$.
  Then,
  $g \in \ch_i^t\rnd{\sigma_i}$.
  \item
  Because
  $g \in \sh_i^s\rnd{h}$,
  and
  $s \geq t$,
  $g^t = h^t$.
  \item
  Because
  $g \in \ch_i^t\rnd{\sigma_i}$,
  and
  $h \in \ch_i^S\rnd{\sigma_i}$,
  $g,h \in \ch_i^{\sqr{t,s} \cap S}\rnd{\sigma_i}$;
  because $g^s = h^s$,
  $g^{\sqrrnd{t,s} \cap S} = f^{\sqrrnd{t,s} \cap S}$.
  Thus,
  by Lemma \ref{lem:increment_of_uniqueness},
  $g_i^{\sqr{t,s} \cap S} = h_i^{\sqr{t,s} \cap S}$.
  Note that
  $S \subset \sqr{t,s}$.
  Then,
  $g_i^S = h_i^S$.
 \end{itemize}
\end{proof}

\begin{lemma}\label{lem:initial_uniqueness}
 Let $h,g \in \ch_i^S\rnd{\sigma_i}$ such that
 $h^t = g^t$.
 Then,
 there exists $s \in T_{> t}$ such that
 $h_i^{\sqrrnd{t,s} \cap S} = g_i^{\sqrrnd{t,s} \cap S}$.
\end{lemma}
\begin{proof}
 By Lemma \ref{lem:extension_of_consistency},
 for any $f \in \crly{h,g}$,
 there exists $f' \in \ch_i^t\rnd{\sigma_i}$ such that
 $\rnd{f'}^t = f^t$,
 and
 $\rnd{f'}_i^S = f_i^S$.
 Note that
 $
  \rnd{h'}^t
  =
  h^t
  =
  g^t
  =
  \rnd{g'}^t
 $.
 Then,
 by Axiom \ref{ax:initial_uniqueness},
 there exists $s \in T_{> t}$ such that
 $\rnd{h'}_i^{\sqrrnd{t,s}} = \rnd{g'}_i^{\sqrrnd{t,s}}$.
 Because $\rnd{h'}_i^{\sqrrnd{t,s}} = \rnd{g'}_i^{\sqrrnd{t,s}}$,
 $\rnd{h'}_i^{\sqrrnd{t,s} \cap S} = \rnd{g'}_i^{\sqrrnd{t,s} \cap S}$;
 for any $f \in \crly{h,g}$,
 because $\rnd{f'}_i^S = f_i^S$,
 $\rnd{f'}_i^{\sqrrnd{t,s} \cap S} = f_i^{\sqrrnd{t,s} \cap S}$.
 Thus,
 $h_i^{\sqrrnd{t,s}\cap S} = g_i^{\sqrrnd{t,s} \cap S}$.
\end{proof}

\begin{lemma}\label{lem:part_uniqueness}
 Let
 $t \in T$,
 $S \in \mathcal T^t$,
 $h \in H$,
 $\sigma \in \tilde\Sigma$
 and
 $g,f \in \sh^t\rnd{h} \cap \ch^S\rnd{\sigma}$.
 Then,
 $g^S = f^S$.
\end{lemma}
\begin{proof}
 By Lemma \ref{lem:totally_ordered_partition},
 for
 any $i \in N$
 and
 any $e \in \crly{g,f}$,
 $\pi_i^t\rnd{e} \in \Pi^t$.
 Let
 $
  \rho
  :=
  \bigcap
  \set
  {
   \pi_i^t\rnd{e}
  }
  {
   i \in N
   \wedge
   e \in \crly{g,f}
  }
 $.
 By
 Lemma \ref{lem:totally_ordered_partition},
 Axiom \ref{ax:well-orderedness}
 and
 Lemma \ref{lem:intersection_of_finite_partitions},
 $\rho \in \hat{\hat\Pi}^t$.
 Let $\tau := \set{R \cap S}{R \in \pi}$.
 Then,
 $\tau \in \hat{\hat\Pi}^S$,
 i.e.,
 $\rnd{\tau,\leq}$ is a well-ordered set.
 By transfinite induction,
 show that
 for any $R \in \tau$,
 $g^R = f^R$.
 Let $R \in \tau$.
 Suppose that
 for any $Q \in \tau$ with $Q < R$,
 $g^Q = f^Q$
 (induction hypothesis).
 Let $s := \inf R$.
 Because $t$ is a lower bound of $R$,
 $s \geq t$;
 there exists $r \in R$,
 and
 $s \leq r$.
 Thus,
 $s \in \sqr{t,r}$.
 Note that
 $t,r \in S$.
 Note also that
 $S$ is a connected set in $T^t$ equipped with the order topology on $\rnd{T^t,\leq}$,
 and
 thus,
 $S$ is an interval in $\rnd{T^t,\leq}$.
 Then,
 $s \in S$.

 Consider the case where
 $T_{> s} = \emptyset$.
 For any $e \in \crly{g,f}$,
 because
 $e \in \ch^S\rnd{\sigma}$,
 and
 $s \in S$,
 $e\rnd{s} = \sigma^s\rnd{e}$;
 because
 by the induction hypothesis,
 $g^s = f^s$,
 $\sigma^s\rnd{g} = \sigma^s\rnd{f}$.
 Hence,
 $g\rnd{s} = f\rnd{s}$.
 Thus,
 $g^{\crly{s}} = f^{\crly{s}}$.
 Note that
 because $T_{> s} = \emptyset$,
 $R = \crly{s}$.
 Then,
 $g^R = f^R$.

 Consider the case where
 $T_{> s} \neq \emptyset$.
 Let $Q := \crly{s} \cup R$.
 $Q \in \mathcal T^s$;
 because
 $g,f \in \ch^S\rnd{\sigma}$,
 and
 $Q = \crly{s} \cup R \subset S$,
 for any $i \in N$,
 $g,f \in \ch_i^Q\rnd{\sigma_i}$;
 by the induction hypothesis,
 $g^s = f^s$.
 Thus,
 by Lemma \ref{lem:initial_uniqueness},
 for any $i \in N$,
 there exists $r_i \in T_{> s}$ such that
 $g_i^{\sqrrnd{s,r_i} \cap Q} = f_i^{\sqrrnd{s,r_i} \cap Q}$.
 Let $r := \min_{i \in N} r_i$.
 Then,
 $g^{\sqrrnd{s,r} \cap Q} = f^{\sqrrnd{s,r} \cap Q}$.
 By Lemma \ref{lem:increment_of_uniqueness'},
 $g^{\sqr{s,r} \cap Q} = f^{\sqr{s,r} \cap Q}$.
 There exists $q \in \sqr{s,r} \cap R$.
 Because $q \in \sqr{s,r} \cap R \subset \sqr{s,r} \cap Q$,
 $g\rnd{q} = f\rnd{q}$;
 for
 any $i \in N$
 and
 any $e \in \crly{g,f}$,
 because there exists $P \in \pi_i^t\rnd{e}$ such that $P \supset R$,
 for any $p \in R$,
 $e_i\rnd{p} = e_i\rnd{q}$.
 Thus,
 for any $p \in R$,
 $g\rnd{p} = f\rnd{p}$,
 i.e.,
 $g^R = f^R$.

 Thus,
 for any $R \in \tau$,
 $g^R = f^R$.
 Hence,
 $g^S = f^S$.
\end{proof}

Let
$t \in T$,
$h \in H$
and
$\sigma \in \tilde \Sigma$.

\paragraph{Existence}

Let
$
 S
 :=
 \set
 {
  s \in T_{\geq t}
 }
 {
  \sh^t\rnd{h} \cap \ch^{\sqr{t,s}}\rnd{\sigma}
  \neq
  \emptyset
 }
$.
For any $s \in S$,
there exists $g_s \in \sh^t\rnd{h} \cap \ch^{\sqr{t,s}}\rnd{\sigma}$.
There exists $g \in \sh^t\rnd{h}$ such that
for any $s \in S$,
$g\rnd{s} = g_s\rnd{s}$.
Let $s \in S$.
Let $r \in \sqrrnd{t,s}$.
Then,
$g_s \in \sh^t\rnd{h} \cap \ch^{\sqr{t,r}}\rnd{\sigma}$.
Thus,
Lemma \ref{lem:part_uniqueness},
$g_r\rnd{r} = g_s\rnd{r}$.
Note that
$g_r\rnd{r} = g\rnd{r}$.
Then,
$g\rnd{r} = g_s\rnd{r}$.
Thus,
$g^s = \rnd{g_s}^s$.
Hence,
$
 g\rnd{s}
 =
 g_s\rnd{s}
 =
 \sigma^s\rnd{g_s}
 =
 \sigma^s\rnd{g}
$.
Thus,
$g \in \sh^t\rnd{h} \cap \sh^S\rnd{\sigma}$.
Hence,
it suffices to show that
$S = T_{\geq t}$.
Suppose that
$S \neq T_{\geq t}$
(assumption for contradiction).
By the following,
$S$ is nonempty and bounded from above.
\begin{itemize}
 \item
 There exists $f \in \sh^t\rnd{h}$ such that
 $f\rnd{t} = \sigma^t\rnd{h}$.
 Because $f^t = h^t$,
 $
  f\rnd{t}
  =
  \sigma^t\rnd{h}
  =
  \sigma^t\rnd{f}
 $.
 Thus,
 $f \in \ch^{\sqr{t,t}}\rnd{\sigma}$.
 Hence,
 $t \in S$.
 Thus,
 $S \neq \emptyset$.
 \item
 There exists $s \in T_{\geq t} \setminus S$.
 Let $r \in S$.
 Suppose that
 $s \leq r$
 (assumption for contradiction).
 There exists $f \in \sh^t\rnd{h} \cap \ch^{\sqr{t,r}}\rnd{\sigma}$.
 Because $s \leq r$,
 $f \in \sh^t\rnd{h} \cap \ch^{\sqr{t,s}}\rnd{\sigma}$.
 Thus,
 $\sh^t\rnd{h} \cap \ch^{\sqr{t,s}}\rnd{\sigma} \neq \emptyset$.
 Hence,
 $s \in S$,
 which contradicts that $s \notin S$.
 Thus,
 $r < s$.
 Thus,
 $S$ is bounded from above.
\end{itemize}
Let $s := \sup S$.
Let $r \in \sqrrnd{t,s}$.
Because $s = \sup S$,
there exists $q \in \rnd{r,s} \cap S$.
$\sh^t\rnd{h} \cap \ch^{\sqr{t,q}}\rnd{\sigma} \neq \emptyset$.
Note that
because $r < q$,
$\ch^{\sqr{t,r}}\rnd{\sigma} \supset \ch^{\sqr{t,q}}\rnd{\sigma}$.
Then,
$\sh^t\rnd{h} \cap \ch^{\sqr{t,r}}\rnd{\sigma} \neq \emptyset$.
Thus,
$r \in S$.
Hence,
$\sqrrnd{t,s} \subset S$.

Consider the case where
$s \in S$.
By Axiom \ref{ax:traceability},
for any $i \in N$,
there exists $g^i \in \sh_i^s\rnd{g} \cap \ch_i^s\rnd{\sigma_i}$.
There exists $f \in \sh^s\rnd{g}$ such that
for any $i \in N$,
$f_i = \rnd{g^i}_i$.
By Axiom \ref{ax:traceability},
for any $i \in N$,
there exists $f^i \in \sh_i^s\rnd{f} \cap \ch_i^s\rnd{\sigma_i}$.
Thus,
by Axiom \ref{ax:initial_uniqueness},
for any $i \in N$,
there exists $r_i \in T_{> s}$ such that
$\rnd{g^i}_i^{\sqrrnd{s,r_i}} = \rnd{f^i}_i^{\sqrrnd{s,r_i}}$.
Let $r := \min_{i \in N} r_i$.
Then,
for
any $i \in N$
and
any $q \in \sqrrnd{s,r}$,
$
 f_i\rnd{q}
 =
 \rnd{g^i}_i\rnd{q}
 =
 \rnd{f^i}_i\rnd{q}
$,
and
$f_{-i}\rnd{q} = \rnd{f^i}_{-i}\rnd{q}$;
thus,
$f\rnd{q} = f^i\rnd{q}$.
Hence,
for
any $i \in N$
and
any $q \in \sqrrnd{s,r}$,
$
 f_i\rnd{q}
 =
 \rnd{f^i}_i\rnd{q}
 =
 \sigma_i^q\rnd{f^i}
 =
 \sigma_i^q\rnd{f}
$.
Note that
because
$f^s = g^s$,
and
$
 g
 \in
 \sh^t\rnd{h} \cap \ch^S\rnd{\sigma}
 \subset
 \sh^t\rnd{h} \cap \ch^{\sqrrnd{t,s}}\rnd{\sigma}
$,
$
 f^t
 =
 g^t
 =
 h^t
$,
and
for any $q \in \sqrrnd{t,s}$,
$
 f\rnd{q}
 =
 g\rnd{q}
 =
 \sigma^q\rnd{g}
 =
 \sigma^q\rnd{f}
$.
Then,
$f \in \sh^t\rnd{h} \cap \ch^{\sqrrnd{t,r}}\rnd{\sigma}$.
Thus,
by Lemma \ref{lem:increment_of_existence'},
there exists $\sh^t\rnd{h} \cap \ch^{\sqr{t,r}}\rnd{\sigma} \neq \emptyset$.
Hence,
$r \in S$,
which contradicts that $r > s$.

Consider the case where
$s \notin S$.
$
 g
 \in
 \sh^t\rnd{h} \cap \ch^S\rnd{\sigma}
 \subset
 \sh^t\rnd{h} \cap \ch^{\sqrrnd{t,s}}\rnd{\sigma}
$.
Thus,
by Lemma \ref{lem:increment_of_existence'},
$\sh^t\rnd{h} \cap \ch^{\sqr{t,s}}\rnd{\sigma} \neq \emptyset$.
Hence,
$s \in S$,
which contradicts that $s \notin S$.

\paragraph{Uniqueness}

Let $g,f \in \sh^t\rnd{h} \cap \ch^t\rnd{\sigma}$.
By Lemma \ref{lem:part_uniqueness},
because $g,f \in \sh^t\rnd{h} \cap \ch^{T_{\geq t}}\rnd{\sigma}$,
$g^{T_{\geq t}} = f^{T_{\geq t}}$.
Thus,
because $g^t = h^t = f^t$,
$g = f$.
\qed

\section{Proof of Proposition \ref{prop:well-ordered_time}}\label{sec:proof_of_well-ordered_time}

\paragraph{Satisfaction of Axiom \ref{ax:traceability}}

Let
$t \in T$
and
$h \in H$.
Because $\rnd{T_{\geq t},\leq}$ is a well-ordered set,
by transfinite induction,
define $g \in \sh_i^t\rnd{h}$ as
for any $s \in T_{\geq t}$,
$g_i\rnd{s} = \sigma_i^s\rnd{f}$,
where $f \in \sh_i^t\rnd{h}$ such that
for any $r \in \sqrrnd{t,s}$,
$f_i\rnd{r} = g_i\rnd{r}$.
By the definition of $g$,
$g \in \sh_i^t\rnd{h} \cap \ch_i^t\rnd{\sigma_i}$.

\paragraph{Satisfaction of Axiom \ref{ax:well-orderedness}}
and 
Let
$t \in T$
and
$h \in \ch_i^t\rnd{\sigma_i}$.
Let $\rho \subset 2^{\pi_i^t\rnd{h}} \setminus \crly{\emptyset}$.
Because $\rnd{T_{\geq t},\leq}$ is a well-ordered set,
there exists a minimum $s$ of $\bigcup \rho$.
There exists $S \in \rho$ such that $S \ni s$.
Let $R \in \rho \setminus \crly{S}$.
There exists $r \in R$.
Because
$s$ is a minimum of $\bigcup \rho$,
and
$s \neq r$,
$s < r$.
Note that
by Lemma \ref{lem:totally_ordered_partition},
$S < R$,
or
$R < S$.
Then,
$S < R$.
Thus,
$S$ is a minimum of $\rho$.
Thus,
$\rnd{\pi_i^t\rnd{h},\leq}$ is a well-ordered set.

\paragraph{Satisfaction of Axiom \ref{ax:initial_uniqueness}}

Let
$t \in T$ such that
$T_{> t} \neq \emptyset$
and
$h,g \in \ch_i^t\rnd{\sigma_i}$ such that
$h^t = g^t$.
Because $\sqrrnd{t,t} = \emptyset$,
$h^{\sqrrnd{t,t}} = g^{\sqrrnd{t,t}}$.
Thus,
by Lemma \ref{lem:increment_of_uniqueness},
$h_i^{\sqr{t,t}} = g_i^{\sqr{t,t}}$.
Because $\rnd{T_{\geq t},\leq}$ is a well-ordered set,
There exists a minimum $s$ of $T_{> t}$.
Because $\sqrrnd{t,s} = \crly{t} = \sqr{t,t}$,
$h_i^{\sqrrnd{t,s}} = g_i^{\sqrrnd{t,s}}$.
\qed

\section{Proof of Proposition \ref{prop:inertiality}}\label{sec:proof_of_inertiality}

Suppose that
$\sigma_i$ satisfies Axiom \ref{ax:inertiality}.

\paragraph{Satisfaction of Axiom \ref{ax:traceability}}

Let
$t \in T$
and
$h \in H$.

Show that
\begin{align}\label{eq:part_uniqueness_in_inertiality}
 \rnd{\forall s \in T_{\geq t}}
 \rnd{\forall g,f \in \sh_i^t\rnd{h} \cap \ch_i^{\sqr{t,s}}\rnd{\sigma_i}}
 \rnd{g^{\sqr{t,s}} = f^{\sqr{t,s}}}.
\end{align}
Let
$s \in T_{\geq t}$
and
$g,f \in \sh_i^t\rnd{h} \cap \ch_i^{\sqr{t,s}}\rnd{\sigma_i}$.
Let $S := \set{r \in \sqr{t,s}}{g_i\rnd{r} \neq f_i\rnd{r}}$.
Suppose that
$S \neq \emptyset$
(assumption for contradiction).
Let $r := \inf S$.
\begin{itemize}
 \item
 Consider the case where
 $T_{> r} = \emptyset$.
 Then,
 $S = \crly{r}$,
 and
 $r = t$.
 Thus,
 because
 $g,f \in \sh_i^t\rnd{h} \cap \ch_i^{\sqr{t,s}}\rnd{\sigma_i}$,
 $
  g_i\rnd{r}
  =
  g_i\rnd{t}
  =
  \sigma_i^t\rnd{g}
  =
  \sigma_i^t\rnd{f}
  =
  f_i\rnd{t}
  =
  f_i\rnd{r}
 $.
 Thus,
 $r \notin S$,
 which contradicts that
 $r \in S$.
 \item
 Consider the case where
 $T_{> r} \neq \emptyset$.
 Because
 $r = \inf S$,
 and
 $g,f \in \sh_i^t\rnd{h}$,
 $g_i^r = f_i^r$;
 because
 $g,f \in \sh_i^t\rnd{h}$,
 $
  g_{-i}
  =
  h_{-i}
  =
  f_{-i}
 $.
 Thus,
 $g^r = f^r$.
 Hence,
 by Axiom \ref{ax:inertiality},
 there exist $q \in T_{> r}$ and $a_i \in A_i$ such that
 for any $p \in \sqrrnd{r,q}$,
 $
  \sigma_i^p\rnd{g}
  =
  a_i
  =
  \sigma_i^p\rnd{f}
 $.
 \begin{itemize}
  \item
  Consider the subcase where
  $q \leq s$.
  Because $g,f \in \ch_i^{\sqr{t,s}}\rnd{\sigma_i}$,
  for any $p \in \sqrrnd{r,q}$,
  $
   g_i\rnd{p}
   =
   \sigma_i^p\rnd{g}
   =
   \sigma_i^p\rnd{f}
   =
   f_i\rnd{p}
  $.
  Thus,
  $\sqrrnd{r,q} \cap S = \emptyset$,
  which contradicts that
  $r = \inf S$.
  \item
  Consider the subcase where
  $q > s$.
  Because $g,f \in \ch_i^{\sqr{t,s}}\rnd{\sigma_i}$,
  for any $p \in \sqr{r,s}$,
  $
   g_i\rnd{p}
   =
   \sigma_i^p\rnd{g}
   =
   \sigma_i^p\rnd{f}
   =
   f_i\rnd{p}
  $.
  Thus,
  because $g_i^r = f_i^r$,
  $S = \emptyset$,
  which contradicts that
  $S \neq \emptyset$.
 \end{itemize}
\end{itemize}
Hence,
$S = \emptyset$.
Thus,
$g_i^{\sqr{t,s}} = f_i^{\sqr{t,s}}$.
Note that
$g,f \in \sh_i^t\rnd{h}$,
and
thus,
$
 g_{-i}
 =
 h_{-i}
 =
 f_{-i}
$.
Then,
$g^{\sqr{t,s}} = f^{\sqr{t,s}}$.

Let $S := \set{s \in T_{\geq t}}{\sh_i^t\rnd{h} \cap \ch_i^{\sqr{t,s}}\rnd{\sigma_i} \neq \emptyset}$.
For any $s \in T_{\geq t}$,
if there exists $r \in S$ such that
$r \geq s$,
there exists $g \in \sh_i^t\rnd{h} \cap \ch_i^{\sqr{t,r}}\rnd{\sigma_i}$;
$g \in \sh_i^t\rnd{h} \cap \ch_i^{\sqr{t,s}}\rnd{\sigma_i}$;
thus,
$\sh_i^t\rnd{h} \cap \ch_i^{\sqr{t,s}}\rnd{\sigma_i} \neq \emptyset$;
hence,
$s \in S$.
Thus,
\begin{align}\label{eq:existence_of_consistent_history}
 \rnd{\forall s \in T_{\geq t}}
 \rnd
 {
  \rnd
  {
   \rnd{\exists r \in S}
   \rnd{r \geq s}
  }
  \rightarrow
  \rnd{s \in S}
 }.
\end{align}
For any $s \in S$,
there exists $g_s \in \sh_i^t\rnd{h} \cap \ch_i^{\sqr{t,s}}\rnd{\sigma_i}$.
There exists $g \in \sh_i^t\rnd{h}$ such that
for any $s \in S$,
$g\rnd{s} = g_s\rnd{s}$.
Let $s \in S$.
Let $r \in \sqrrnd{t,s}$.
By (\ref{eq:existence_of_consistent_history}),
$r \in S$.
$g_r,g_s \in \sh_i^t\rnd{h} \cap \ch_i^{\sqr{t,r}}\rnd{\sigma_i}$.
Thus,
by (\ref{eq:part_uniqueness_in_inertiality}),
$g_r\rnd{r} = g_s\rnd{r}$.
Note that
$g_r\rnd{r} = g\rnd{r}$.
Then,
$g\rnd{r} = g_s\rnd{r}$.
Thus,
$g^s = \rnd{g_s}^s$.
Hence,
$\sigma_i^s\rnd{g} = \sigma_i^s\rnd{g_s}$;
because $g\rnd{s} = g_s\rnd{s}$,
$g_i\rnd{s} = \rnd{g_s}_i\rnd{s}$;
because $g_s \in \ch_i^{\sqr{t,s}}\rnd{\sigma_i}$,
$\rnd{g_s}_i\rnd{s} = \sigma_i^s\rnd{g_s}$.
Thus,
$
 g_i\rnd{s}
 =
 \sigma_i^s\rnd{g}
$.
Thus,
\begin{align}\label{eq:consistency_of_diagonal_history}
 g \in \sh_i^t\rnd{h} \cap \ch_i^S\rnd{\sigma_i}.
\end{align}
It suffices to show that
$S = T_{\geq t}$.
Suppose that
$S \neq T_{\geq t}$
(assumption for contradiction).
By the following,
$S$ is nonempty and bounded from above.
\begin{itemize}
 \item
 There exists $f \in \sh_i^t\rnd{h}$ such that
 $f_i\rnd{t} = \sigma_i^t\rnd{h}$.
 Because $f^t = h^t$,
 $
  f_i\rnd{t}
  =
  \sigma_i^t\rnd{h}
  =
  \sigma_i^t\rnd{f}
 $.
 Thus, 
 $f \in \ch_i^{\sqr{t,t}}\rnd{\sigma_i}$.
 Hence,
 $\sh_i^t\rnd{h} \cap \ch_i^{\sqr{t,t}}\rnd{\sigma_i} \neq \emptyset$.
 Thus,
 $t \in S$.
 Hence,
 $S \neq \emptyset$.
 \item
 Because $S \neq T_{\geq t}$
 there exists $s \in T_{\geq t} \setminus S$.
 By (\ref{eq:existence_of_consistent_history}),
 for any $r \in S$,
 $r < s$.
 Thus,
 $S$ is bounded from above.
\end{itemize}
Let $s := \sup S$.
In the following two cases,
a contradiction occurs.
Thus,
$S = T_{\geq t}$.

Consider the case where
$T_{> s} = \emptyset$.
There exists $f \in \sh_i^s\rnd{g}$ such that
$f_i^s\rnd{s} = \sigma_i^s\rnd{g}$.
Let $r \in \sqrrnd{t,s}$.
Because $s = \sup S$,
there exists $q \in \rndsqr{r,s} \cap S$.
By (\ref{eq:existence_of_consistent_history}),
$r \in S$.
Thus,
by $f^s = g^s$ and (\ref{eq:consistency_of_diagonal_history}),
$
 f_i\rnd{r}
 =
 g_i\rnd{r}
 =
 \sigma_i^r\rnd{g}
 =
 \sigma_i^r\rnd{f}
$.
Because $f^s = g^s$,
$
 f_i\rnd{s}
 =
 \sigma_i^s\rnd{g}
 =
 \sigma_i^s\rnd{f}
$.
Thus,
$f \in \ch_i^{\sqr{t,s}}\rnd{\sigma_i}$.
Because
$f \in \sh_i^s\rnd{g}$,
and
$g \in \sh_i^t\rnd{h}$,
$f \in \sh_i^t\rnd{h}$.
Thus,
$\sh_i^t\rnd{h} \cap \ch_i^{\sqr{t,s}}\rnd{\sigma_i} \neq \emptyset$.
Thus,
$s \in S$.
For any $r \in T_{\geq t}$,
by $s \in S$, $s \geq r$ and (\ref{eq:existence_of_consistent_history}),
$r \in S$.
Thus,
$S = T_{\geq t}$,
which contradicts that
$S \neq T_{\geq t}$.

Consider the case where
$T_{> s} \neq \emptyset$.
By Axiom \ref{ax:inertiality},
there exist $r \in T_{> s}$ and $a_i \in A_i$ such that
\begin{align}\label{eq:inertiality}
 \rnd{\forall q \in \sqrrnd{s,r}}
 \rnd{\forall f \in \sh_i^s\rnd{g}}
 \rnd{\sigma_i^q\rnd{f} = a_i}.
\end{align}
Let $f \in \sh_i^s\rnd{g}$ such that
for any $q \in T_{\geq s}$,
$f_i\rnd{q} = a_i$.
For any $q \in \sqrrnd{t,s}$,
by the definition of $f$,
$f^s = g^s$;
by $q < \sup S$ and (\ref{eq:existence_of_consistent_history}),
$q \in S$;
by (\ref{eq:consistency_of_diagonal_history}),
$g \in \ch_i^S\rnd{\sigma_i}$;
by $f^s = g^s$ and $q < s$,
$f^q = g^q$;
thus,
$
 f_i\rnd{q}
 =
 g_i\rnd{q}
 =
 \sigma_i^q\rnd{g}
 =
 \sigma_i^q\rnd{f}
$.
For any $q \in \sqrrnd{s,r}$,
by the definition of $f$ and (\ref{eq:inertiality}),
$
 f_i\rnd{q}
 =
 a_i
 =
 \sigma_i^q\rnd{f}
$.
Thus,
$f \in \ch_i^{\sqrrnd{t,r}}\rnd{\sigma_i}$.
By the definition of $f$ and $t \leq s$,
$f \in \sh_i^t\rnd{g}$;
thus,
by $g \in \sh_i^t\rnd{h}$,
$f \in \sh_i^t\rnd{h}$.
Hence,
$f \in \sh_i^t\rnd{h} \cap \ch_i^{\sqrrnd{t,r}}\rnd{\sigma_i}$.
Thus,
by Lemma \ref{lem:increment_of_existence},
$\sh_i^t\rnd{h} \cap \ch_i^{\sqr{t,r}} \neq \emptyset$.
Hence,
$r \in S$.
Thus,
$r \leq s$,
which contradicts that
$r > s$.

\paragraph{Satisfaction of Axiom \ref{ax:well-orderedness}}

Let
$t \in T$
and
$h \in \ch_i^t\rnd{\sigma_i}$.
Let $\rho \in 2^{\pi_i^t\rnd{h}} \setminus \crly{\emptyset}$.
Let $s := \inf \bigcup \rho$.
In the following two cases,
there exists a minimum of $\rho$.
Thus,
$\rnd{\pi_i^t\rnd{h},\leq}$ is a well-ordered set.

Consider the case where
$T_{> s} = \emptyset$.
Then,
$\bigcup \rho = \crly{s}$.
Thus,
$\rho = \crly{\crly{s}}$.
Hence,
$\crly{s}$ is a minimum of $\rho$.

Consider the case where
$T_{> s} \neq \emptyset$.
By Axiom \ref{ax:inertiality},
there exist $r \in T_{> s}$ and $a_i \in A_i$ such that
\begin{align}
 \rnd{\forall q \in \sqrrnd{s,r}}
 \rnd
 {
  h_i\rnd{q}
  =
  \sigma_i^q\rnd{h}
  =
  a_i
 }.\label{eq:inertiality}
\end{align}
Because $s = \inf \bigcup \rho$,
there exists $q \in \sqrrnd{s,r} \cap \bigcup \rho$.
There exists $S \in \rho$ such that
$S \ni q$.
Suppose that
$\neg \rnd{\sqrrnd{s,r} \subset S}$
(assumption for contradiction).
Because
$S,\sqrrnd{s,r} \in \mathcal C^t$,
and
$S \cap \sqrrnd{s,r} \ni q$,
$S \cup \sqrrnd{s,r} \in \mathcal C^t$;
by the assumption for contradiction,
$S \cup \sqrrnd{s,r} \supsetneq S$;
thus,
because $S \in \pi_i^t\rnd{h}$,
$\card{h_i\rnd{S \cup \sqrrnd{s,r}}} \neq 1$.
Because
$S \in \pi_i^t\rnd{h}$,
and
$S \ni q$,
$h_i\rnd{S} = \crly{h_i\rnd{q}}$;
by (\ref{eq:inertiality}) and $\sqrrnd{s,r} \ni q$,
$h_i\rnd{\sqrrnd{s,r}} = \crly{h_i\rnd{q}}$;
thus,
$\card{h_i\rnd{S \cup \sqrrnd{s,r}}} = 1$,
which contradicts that
$\card{h_i\rnd{S \cup \sqrrnd{s,r}}} \neq 1$.
Thus,
$\sqrrnd{s,r} \subset S$.
Hence,
$s \in S$.
Let $R \in \rho \setminus \crly{S}$.
There exists $p \in R$.
Because
$s = \inf \bigcup \rho$,
and
$s \neq p$,
$s < p$.
Note that
by Lemma \ref{lem:totally_ordered_partition},
$S < R$,
or
$R < S$.
Then,
$S < R$.
Thus,
$S$ is a minimum of $\rho$.

\paragraph{Satisfaction of Axiom \ref{ax:initial_uniqueness}}

Let $t \in T$ such that
$T_{> t} \neq \emptyset$
and
$h,g \in \ch_i^t\rnd{\sigma_i}$ such that
$h^t = g^t$.
Then,
by Axiom \ref{ax:inertiality},
there exist $s \in T_{> t}$ and $a_i \in A_i$ such that
for any $r \in \sqrrnd{t,s}$,
$
 h_i\rnd{r}
 =
 \sigma_i^r\rnd{h}
 =
 a_i
 =
 \sigma_i^r\rnd{g}
 =
 g_i\rnd{r}
$.
Thus,
$h_i^{\sqrrnd{t,s}} = g_i^{\sqrrnd{t,s}}$.
\qed

\section{Proof of Proposition \ref{prop:frictionality}}\label{sec:proof_of_frictionality}

Suppose that
$\sigma_i$ satisfies Axiom \ref{ax:frictionality}.

\paragraph{Satisfaction of Axiom \ref{ax:well-orderedness}}

Let
$t \in T$
and
$h \in \ch_i^t\rnd{\sigma_i}$.
Let $\rho \in 2^{\pi_i^t\rnd{h}} \setminus \crly{\emptyset}$.
Let $s := \inf \bigcup \rho$.
In the following two cases,
there exists a minimum of $\rho$.
Thus,
$\rnd{\pi_i^t\rnd{h},\leq}$ is a well-ordered set.

Consider the case where
$s \in \bigcup \rho$.
There exists $S \in \rho$ such that $S \ni s$.
Let $R \in \rho \setminus \crly{S}$.
There exists $r \in R$.
Because
$s = \inf \bigcup \rho$,
and
$s \neq r$,
$s < r$.
Note that
by Lemma \ref{lem:totally_ordered_partition},
$S < R$,
or
$R < S$.
Then,
$S < R$.
Thus,
$S$ is a minimum of $\rho$.

Consider the case where
$s \notin \bigcup \rho$.
There exists $r \in \bigcup \rho$.
Because $s \notin \bigcup \rho$,
$r > s$.
Let $S := \set{q \in \rnd{s,r}}{h_i\rnd{q} \neq z_i}$.
If $S = \emptyset$,
$\card{h_i\rnd{\rnd{s,r}}} \leq 1$;
if $S \neq \emptyset$,
by Axiom \ref{ax:frictionality},
there exists a minimum $q$ of $S$,
and
$\card{h_i\rnd{\rnd{s,q}}} \leq 1$.
Thus,
there exists $r \in T_{> s}$ such that
$\card{h_i\rnd{\rnd{s,r}}} \leq 1$.
Because
$s = \inf \bigcup \rho$,
and
$s \notin \bigcup \rho$,
there exists $q \in \rnd{s,r} \cap \rnd{\bigcap \rho}$.
Thus,
$\rnd{s,r} \neq \emptyset$,
and
hence,
$\card{h_i\rnd{\rnd{s,r}}} = 1$.
There exists $S \in \rho$ such that
$S \ni q$.
Let
$R \in \rho \setminus \crly{S}$.
There exists $p \in R$.
Because $S,\rnd{s,r} \in \mathcal C^t$,
$S \cup \rnd{s,r} \in \mathcal C^t$;
because
$\card{h_i\rnd{S}} = \card{h_i\rnd{\rnd{s,r}}} = 1$,
and
$S \cap \rnd{s,r} \neq \emptyset$,
$\card{h_i\rnd{S \cup \rnd{s,r}}} = 1$.
Thus,
by $S \in \pi_i^t\rnd{h}$ and the definition of $\pi_i^t\rnd{h}$,
$S \cup \rnd{s,r} = S$.
Hence,
$S \supset \rnd{s,r}$.
Thus,
$p \notin \rnd{s,r}$.
Hence,
$r \leq p$.
Note that
$q \in \rnd{s,r}$,
and
thus,
$q < r$.
Then,
$q < p$.
Note that
by Lemma \ref{lem:totally_ordered_partition},
$S < R$,
or
$R < S$.
Then,
$S < R$.
Thus,
$S$ is a minimum of $\rho$.

\paragraph{Satisfaction of Axiom \ref{ax:initial_uniqueness}}

Let
$t \in T$ such that
$T_{> t} \neq \emptyset$
and
$h,g \in \ch_i^t\rnd{\sigma_i}$ such that
$h^t = g^t$.
There exists $s \in T_{> t}$.
Let
$
 S
 :=
 \set
 {
  r \in \rnd{t,s}
 }
 {
  h_i\rnd{r} \neq z_i
  \vee
  g_i\rnd{r} \neq z_i
 }
$.
By Lemma \ref{lem:increment_of_uniqueness'},
$h_i\rnd{t} = g_i\rnd{t}$.
Thus,
it suffices to show that
there exists $r \in T_{> t}$ such that
for any $q \in \rnd{t,r}$,
$h_i\rnd{q} = g_i\rnd{q}$.
   
Consider the case where
$S = \emptyset$.
Then,
for any $q \in \rnd{t,s}$,
$h_i\rnd{q} = z_i = g_i\rnd{q}$.

Consider the case where
$S \neq \emptyset$.
Then,
by Axiom \ref{ax:frictionality},
there exists a minimum $r$ of $S$.
Thus,
for any $q \in \rnd{t,r}$,
$h_i\rnd{q} = z_i = g_i\rnd{q}$.
\qed

\newpage
\bibliographystyle{plainnat}
\bibliography{Strategies_in_Totally_Ordered_Times_1}

\begin{thebibliography}{2}
\providecommand{\natexlab}[1]{#1}
\providecommand{\url}[1]{\texttt{#1}}
\expandafter\ifx\csname urlstyle\endcsname\relax
  \providecommand{\doi}[1]{doi: #1}\else
  \providecommand{\doi}{doi: \begingroup \urlstyle{rm}\Url}\fi

\bibitem[Bergin and MacLeod(1993)]{Bergin1993}
J.~Bergin and W.~B. MacLeod.
\newblock Continuous time repeated games.
\newblock \emph{International Economic Review}, 34:\penalty0 21--37, 1993.

\bibitem[Kamada and Rao(2021)]{Kamada2021}
Y.~Kamada and N.~Rao.
\newblock Strategies in stochastic continuous-time games.
\newblock Unpublished manuscript, 2021.

\end{thebibliography}

\end{document}